\documentclass[11pt]{amsart}
\usepackage[a4paper]{geometry}

\usepackage{mathrsfs}
\usepackage{amsfonts}
\usepackage{amsmath} 
\usepackage{amssymb} 
\usepackage{amsthm}
\usepackage{graphicx} 
\usepackage{url}
\usepackage{color}
\usepackage{hyperref}
\definecolor{refcolor}{RGB}{0,0,190}
\hypersetup{
    colorlinks,
    citecolor=refcolor,
    filecolor=refcolor,
    linkcolor=refcolor,
    urlcolor=refcolor
}

\usepackage{bm}

\theoremstyle{definition}
\newtheorem{definition}{Definition}[section]
\newtheorem{example}[definition]{Example}

\theoremstyle{plain}

\newtheorem{lemma}[definition]{Lemma}

\title{On the wavefunction collapse*}
\author{Ovidiu Cristinel Stoica**}
\thanks{*Published in Quanta, DOI: \href{http://dx.doi.org/10.12743/quanta.v5i1.40}{http://dx.doi.org/10.12743/quanta.v5i1.40}
\\ **Department of Theoretical Physics, National Institute of Physics and Nuclear Engineering -- Horia Hulubei, Bucharest, Romania. Email: \href{mailto:cristi.stoica@theory.nipne.ro}{cristi.stoica@theory.nipne.ro},  \href{mailto:holotronix@gmail.com}{holotronix@gmail.com}}
\date{\today}

\begin{document}

\begin{abstract}
Wavefunction collapse is usually seen as a discontinuous violation of the unitary evolution of a quantum system, caused by the observation. Moreover, the collapse appears to be nonlocal in a sense which seems at odds with General Relativity. In this article the possibility that the wavefunction evolves continuously and hopefully unitarily during the measurement process is analyzed. It is argued that such a solution has to be formulated using a time symmetric replacement of the initial value problem in Quantum Mechanics. Major difficulties in apparent conflict with unitary evolution are identified, but eventually its possibility is not completely ruled out. This interpretation is in a weakened sense both local and realistic, without contradicting Bell's theorem. Moreover, if it is true, it makes Quantum Mechanics consistent with General Relativity in the semiclassical framework.
\end{abstract}

\maketitle

\section{Introduction}
\label{s:intro}

\subsection{Unitary evolution and wavefunction collapse}
\label{s:standardQM}

The state of a quantum system is represented by a vector $|{\psi}\rangle$ in a Hilbert space $\mathcal{H}$. Its evolution is governed by the Schr\"odinger equation, 
\begin{equation}
\label{eq_schrod}
\Bigg\{
\begin{array}{ll}
	\imath\hbar\frac{\partial}{\partial t}|{\psi(t)}\rangle=\hat H(t)|{\psi(t)}\rangle \\
	|{\psi(t_a)}\rangle=|{\psi_a}\rangle
\end{array}
\end{equation}
where $\hat H(t)$ is the Hamiltonian, which is a Hermitian operator on $\mathcal{H}$. If the quantum system is closed, then $\hat H$ is time independent.
The solutions of the Schr\"odinger equation have the form
\begin{equation}
\label{eq_unitary_evolution}
|{\psi(t_b)}\rangle=\hat U(t_b,t_a)|{\psi(t_a)}\rangle
\end{equation}
where $\hat U(t_b,t_a)$ is a unitary operator on $\mathcal{H}$, given by:
\begin{equation}
\label{eq_unitary_evolution_operator}
\hat U(t_b,t_a) = \mathcal T\left(e^{-\frac{\imath}{\hbar} \int_{t_a}^{t_b}\hat H(t)\operatorname{d} t}\right)
\end{equation}
where $\mathcal T$ is the \emph{time ordering operator}, needed because the Hamiltonians at different times might not commute.
In the case of time independent $\hat H$,
\begin{equation}
\label{eq_unitary_evolution_operator_time_independent}
\hat U(t_b,t_a)=e^{-\frac{\imath}{\hbar} \left(t_b-t_a\right)\hat H}
\end{equation}

\emph{Observables} are represented by Hermitian operators $\hat {\mathcal O}$ on the Hilbert space $\mathcal{H}$. The outcome of a measurement is an eigenvalue $\lambda\in\mathbb{R}$ of $\hat {\mathcal O}$, and the state of the observed system is an eigenstate $|{\lambda}\rangle$ of $\hat {\mathcal O}$ corresponding to $\lambda$. The probability density that a quantum system previously in the state $|{\psi}\rangle$ is found to be in the eigenstate $|{\lambda}\rangle$ is, according to the \emph{Born rule}, $|\langle{\lambda}|{\psi}\rangle|^2$.
In particular, if $|{\psi}\rangle$ represents a single particle, then according to the \emph{Born rule}, the probability density that the particle is detected at a time $t_a$ at the position $x_a\in\mathbb{R}^3$ is
\begin{equation}
\label{eq_Born_rule_position}
\text{P}_{(x_a,t_a)}=|\langle{x_a}|{\psi(t_a)}\rangle|^2
\end{equation}
where $|{x_a}\rangle$ is the eigenstate of the position operator $\hat x$ corresponding to the position $x_a$, so that $\langle{x}|{x_a}\rangle$ is equal to the Dirac distribution $\delta(x-x_a)$.

Of course, after the particle was detected at the time $t_a$ at the position $x_a$, the probability to find it elsewhere vanishes, so the wavefunction changes --  we say it \emph{collapsed} at the position $x_a$.

The collapse specified by the Born rule suggests that the wavefunction is merely a tool for calculating the probabilities.
On the other hand, Quantum Mechanics (QM) describes everything -- particles, atoms, molecules, hence all material objects -- as wavefunctions, so are they merely probabilistic waves?

The notion of discontinuous collapse has to face some problems. First, how can the Schr\"odinger equation, so successfully confirmed, be accommodated with the apparent wavefunction collapse? How can we reconcile a collapse taking place simultaneously everywhere in space, with Relativity, which does not accept the notion of absolute simultaneity?

On the other hand, trying to replace it with an effect resulting from dynamics also encounters severe difficulties, some of which will be explored here.

\subsection{Motivation}

In this article, I am interested in exploring the possibility that the dynamics of quantum systems, governed by the Schr\"odinger equation, can take place without discontinuous collapse, even during measurements. Hopefully we can find out that it can evolve by Schr\"odinger's equation, but maybe we need a general relativistic version, or at least an approximation like the non-linear Schr\"odinger--Newton equation \cite{RuffiniBonazzola1969SchrodingerNewtonEquation} (non-linear modifications of the Schr\"odinger equation of the type studied by Weinberg are known to be signaling \cite{gisin1990weinbergSignaling}, but it is not excluded that other non-linear approaches may not signal). The literature exploring the possibility that the Schr\"odinger--Newton equation introduces enough non-linearity so that it accounts for collapse in a way similar to the Ghirardi-Rimini-Weber approach \cite{GRW86} is very rich, see for example \cite{diosi1984gravitationQMlocalization,penrose1996gravityQuantumStateReduction}. The approach presented here is different, in the sense that it tries to account for the collapse with the minimal possible departure from the unitary evolution, or at least from a continuous, albeit non-unitary or even non-linear evolution.

The idea that unitary evolution is not broken is central also in the \emph{many worlds interpretation} \cite{Eve57,Eve73,dWEG73,Vaidman2002MWI}, enhanced with the proposal that \emph{decoherence} can resolve the measurement problem \cite{Zeh96,Zurek2003Decoherence,schlosshauer2005decoherence} (although there are some serious objections to this proposal \cite{pessoa1997canDecoherence,leggett2002limitsQM,kastner2014einselection}). But while in these approaches the unitary evolution of Schr\"odinger's equation is maintained at the multiverse level, where all branches are included, the collapse is still present at the level of a branch. Here I am interested whether it is possible to maintain unitary evolution in a single world, or at the branch level.

Given that in standard QM, briefly described in section \S\ref{s:standardQM}, the statistical interpretation of the wavefunction given by Born is confirmed by observations, it is the correct description for all practical purposes. And this description suggests that quantum measurement leads to a discontinuous wavefunction collapse. However, it is still possible that the wavefunction which is inferred from the measurements is not the same as the real wavefunction, and this is one of the central themes of this article.
The reasons which lead to enough flexibility to allow for this possibility are the following:
\begin{enumerate}
	\item 
One cannot measure directly the wavefunction. What the measurements tell us is that the quantum state of the observed system is an eigenstate of the observable.
	\item 
Even this information is subject to inherent limitations given by the \emph{error-disturbance uncertainty relations} \cite{Heisenberg1927Uncertainty,Ozawa2003Heisenberg,Ozawa2013disprovingHeisenbergErrorDisturbance}. 
\end{enumerate}

According to the error-disturbance uncertainty relations, the more precise a measurement is, the more it disturbs the observed state. This means that the collection of results of measurements can only give an approximation of the state of the wavefunction. Given that the constraints imposed by the measurements to the wavefunction are more relaxed than it is usually assumed, a question becomes justified:

\begin{quote}
Is it possible that the real wavefunction can fit the observations provided by measurements, without actually having to collapse in a discontinuous way?
\end{quote}

In standard QM, the wavefunction represents our knowledge about the observed system, and the probabilities of the possible outcomes of future measurements, so let us call the wavefunction representing probabilities \emph{epistemic wavefunction}. What I propose is that one should consider the possibility that there is also a real, \emph{ontic wavefunction}, which evolves continuously even during the collapse, and which is merely approximated by the epistemic wavefunction. 

Our measurements give us the state representing the ontic wavefunction within the limits of error and disturbance. This entails a difference between the real, ontic wavefunction, and our statistical knowledge about it, represented by the epistemic wavefunction. I argue that the collapse we observe takes place only at the epistemic level, while it is still possible that the real wavefunction evolves continuously, following the Schr\"odinger equation, or at least a modified, perhaps non-unitary or even non-linear version.

This proposal is in line with other proposals that there are entities which represent ``things'', ``beables''. This idea is pursued for example in the \emph{de Broglie--Bohm theory} \cite{deBroglie1926OndesEtMouvements,Bohm52,DGZ1996BohmianMechanics,DGZ2012QPhysWithoutQPhil} and other \emph{hidden-variable theories}, and, in a unitary version, in 't Hooft's approach based on cellular automata \cite{tHooft2014CellularAutomatonInterpretationQM}. But the departure of these proposals from Schr\"odinger's equation is significant.
Here I will try to obtain a description still based on wavefunctions, and hopefully still governed by the Schr\"odinger equation. For example the atom is a ``thing'', which contains electrons whose states are very well described by Schr\"odinger's equation, or at least an approximation of it. Schr\"odinger himself originally saw the wavefunctions as real entities, but because entanglement makes them unlike fields or any other classical entities, he did not continue to pursue this possibility.

I will take into account major difficulties encountered by the proposal of a wavefunction which describes reality and not merely probabilities, and see if it survives at the end. If we can obtain a consistent picture, we will be entitled to call this wavefunction \emph{ontic} (and still keep the epistemic, probabilistic approximation, which is the only one we can access by quantum measurements).

I expect that the information obtained from measurements, encoded in the epistemic wavefunction, describes to some degree also the ontic wavefunction. However, if we assume that the measurement also tells the ontic state of the observed system, then the conflict between dynamics and measurement can only be resolved by admitting a discontinuous collapse, either of the kind in the standard QM, or a \emph{spontaneous collapse}, as in the GRW theory \cite{GRW86,Diosi1989UniversalReduction}.

The tension between measurements and unitary evolution seems to lead with necessity to the collapse, so if in reality there is no discontinuous collapse, this can only be achieved if either the measurement or the dynamics is more flexible than we thought (or both). Let us first verify if our assumptions about measurement are true, and only change the dynamics if needed.

The purpose of this exploration is to find the possible conditions that any unitary approach to QM should satisfy, and see whether this possibility is still consistent.

\section{The tension between quantum measurement and the initial conditions of the observed system}
\label{s:qm_init_cond}

It was clear since the dawn of quantum mechanics, especially with von Neumann's formulation \cite{vonNeumann1955foundations}, that the state of the observed system appears to be in general a superposition of the possible results of a measurement, yet at the end of the measurement, the state turns out to be one of these possibilities. This seems to require a projection of the state in an eigenstate of the observable. However, it was suggested that by taking into account the environment, which includes the measurement apparatus, the evolution is still unitary. This idea was developed in the \textit{decoherence program} \cite{Zeh96,Zur98,Zur03a}. Indeed, by accounting for the environment, the density matrix decoheres, so that the off-diagonal terms vanish. The diagonal terms are then interpreted as a statistical ensemble, so we cannot actually claim that the evolution is unitary, because evolving a pure state into a mixture means collapse. In such approaches, unitarity exists only when all decohered histories are taken into account.

In fact, any kind of attempt to give a purely unitary description of the wavefunction collapse at the branch level can work only for very special initial conditions of the observed system and the measurement apparatus. In \cite{Sto12QMb} it was proven that in order to get a unitary description of the measurement process, the initial conditions of the observed system and those of the measurement apparatus have to belong to a set of zero measure of the full Hilbert space. The fact that unitary evolution is compatible with measurement only for special initial initial conditions, requiring therefore a fine tuning, seems to endanger the principle of causality. I will address this delicate problem in  the sections \S\ref{s:delayed_initial_conditions}, \S\ref{s:spacetime_locality} and \S\ref{s:global_consistency}, and provide a more rigorous picture in \S\ref{s:wavefunction-events-picture}.

\section{Propagation of a photon from one place to another}
\label{s:particle_place2place}

Let us start with a simple case, of a photon going from one place to another. The Schr\"odinger equation has to take the wavefunction at the time $t_a$ from the place where it is emitted, and evolve it unitarily to another place, where it is detected at a later time $t_b$. We already see that without a collapse the photon has to have fine tuned initial conditions at $t_a$, so that at $t_b$ is found in a definite place. Let us now see how fine tuned the initial conditions have to be, or if even it is possible for a photon to satisfy both the initial and final conditions.
 
Suppose that a photon is emitted at the position $x_a$ at the time $t_a$, and it is later, at the time $t_b>t_a$, detected at the position $x_b$.
To find the amplitude $\langle{x_b}|{\psi(t_b)}\rangle$, we need to apply the free particle Schr\"odinger equation to the initial state
\begin{equation}
|{\psi(t_a)}\rangle = |{x_a}\rangle
\end{equation}
This makes the momentum completely undetermined, and the photon spreads like a spherical wave, preventing any possibility to be described by a wavefunction evolving unitarily from one point to another. But this is an idealization, because the photon is never emitted from a point-like source. For example, if it is emitted by an atom, the wavefunction $|{\psi(t_a)}\rangle$ is of the size of an atom. But even the wavefunctions of the electrons in the Hydrogen atom extend radially in the entire space. The amplitudes decay exponentially with the distance, but still they do not vanish.
Because of this reason, both the emission and the detection of the photons are wrongfully represented as taking place in a definite position in standard QM. To be realistic we have to admit that what we call the position of emission or detection are actually some average positions, and the true state of a photon when emitted or absorbed is unknown. Using the eigenstates of the position to represent them is an approximation. A better approximation would be a Gaussian function centered at $x_a$ and having the width equal to the radius of the atom.

What if we consider rather than precise locations, more extended regions? Suppose now that immediately after the emission, at the time $t_a$, the photon passes through an extended but bounded region of space $A\subset\mathbb{R}^3$. Given this loosening of the initial condition, could it be possible that the Schr\"odinger equation itself makes the wavefunction of the photon evolve so that at a time $t_b>t_a$ it passes through a bounded region $B\subset\mathbb{R}^3$ (for instance where is detected)? Unfortunately, no matter how large we allow the regions $A$ and $B$ to be, as long as they are bounded, it is still impossible for a free particle to be confined at a time $t_a$ in the region $A$, and later at $t_b$ to region $B$ without breaking the unitary evolution. Suppose that at $t_a$ the \emph{support} of the wavefunction $\psi(t_a)$ is included in $A$, $\operatorname{supp}(|{\psi(t_a)}\rangle)\subseteq A$. Then, its Fourier transform will be an \emph{entire function} (a complex function holomorphic over the entire complex domain), so its support will cover the entire domain of wavelengths. This means that no matter how large is $A$, the wavefunction $|{\psi(t_a)}\rangle$ will be a superposition of plane wavefunctions of almost all possible momenta. So immediately after $t_a$ the wavefunction will spread in the entire space, and there will be no way that at a later time $t_b$, $\operatorname{supp}(|{\psi(t_b)}\rangle)\subseteq B$.

From this point of view, the wavefunction collapse seems like a clean solution to this problem, because it allows the wavefunction extended in the entire space to become suddenly localized in a small region.

Consider now a Hydrogen atom in a water molecule in a glass of water (which is a bounded region). The Hydrogen atom will remain in the glass for long time. If the collapse is the explanation of localization, then in this case the wavefunction of the atom has to collapse all the time, to remain in the glass. An alternative way is to admit that it extends in the entire space, but it is ``more localized'' in a certain position in the glass. This weakening of the condition to be localized at a definite position allows it to remain for long time in the glass, without the need of collapsing all the time.

Gaussian wavepackets, despite saturating the uncertainty relation and remaining Gaussian in the absence of interaction, spread in space. This would make impossible for a small detector on earth to detect without collapse a photon emitted by an atom in a distant galaxy.
A more appropriate solution would require using \emph{nonspreading wavepackets}, so that if they were localized in a place at $t_a$, they will be localized also at $t_b$.
Fortunately, such solutions are known for Dirac, Klein-Gordon and Schr\"odinger equations \cite{courant1966methods,ziolkowski1989space,shaarawi1990novel,Barut1990EequalHW,barut1990quantumParticleLike,hillion1992nondispersive,sheppard2002generalized,zamboni2012soliton}. Moreover, such solutions are even able to reproduce the two-slit interference, and therefore the Born rule for this case \cite{shaarawi1994diffraction}.

This analysis shows that it is not accurate to consider the measurement of position as finding the wavefunction to be precisely localized at a definite position. Rather, ``most of the wavefunction'' is localized in a small region around that position. The epistemic wavefunction is in this case $|{x}\rangle$, because this is what we think we know about the particle detected at $x$, but the ontic wavefunction is something completely different, extended in the entire space, but concentrated around $x$, which would be better approximated by a such a solitonic wavefunction which is ``mostly localized'' around $x$.

Therefore, unitary evolution of photons can accommodate the fact that the photon travels from the place where it is emitted to the place where it is absorbed without breaking unitary evolution. Clearly its wavefunction has to be very special for this, the initial conditions have to be fine tuned to also satisfy the final conditions, as we already know from the discussion in section \S\ref{s:qm_init_cond}.

Note that the notion of ``mostly localized'' does not refer to the probabilities, but to a physical wavefunction whose existence I propose here, which is merely approximated by the measurements. The probabilities apply to our knowledge of the wavefunction, while the localization I am proposing here refers to the physical, ontic wavefunction, whose possibility of existence is explored here.

This only shows that the possibility of this taking place unitarily exists, but it does not explain why the wavefunction takes such a special form. Perhaps a deeper understanding of particles, which still eludes us, will provide an explanation.

\section{The ``unitary collapse'' condition}
\label{s:unitary_collapse_condition}

In order to be rigorous, we would have to define what ``mostly localized'' is. We can define the \emph{degree of localization} of a wavefunction $|{\psi}\rangle$ inside a region $A\subset\mathbb{R}^3$ as
\begin{equation}
\label{eq_degree_localization}
\Lambda_A(|{\psi}\rangle):=\int_{A}\langle{x}|{\psi}\rangle\operatorname{d} x
\end{equation}
Maybe it is more appropriate to use a more elaborate definition, for instance using \emph{standard deviation}. Standard deviation is natural to be used in the case of Gaussian wavepackets, because we can use for example the width of the packet. However, for simplicity we can consider equation \eqref{eq_degree_localization}. Let us fix a value $0<\Lambda \leq 1$ and write
\begin{equation}
\label{eq_localization_region}
|{\psi}\rangle\triangleleft(A,t)
\end{equation}
if $\Lambda_A(|{\psi(t)}\rangle)\geq \Lambda$.
We say that \emph{a particle whose wavefunction is $|{\psi}\rangle$ is in the region $A$ at the time $t$} if $|{\psi}\rangle\triangleleft(A,t)$.

Thus, I propose the following \emph{unitary collapse condition}:
\begin{quote} 
\label{condition:unitary_collapse}
\emph{In the real world, the wavefunction evolves unitarily so that at the times $t_a$ and $t_b$ it passes through the regions $A$ and $B$.}
\end{quote}

In other words, $|{\psi}\rangle$ has to simultaneously satisfy the following three conditions:
\begin{enumerate}
	\item 
$|{\psi}\rangle\triangleleft(A,t_a)$,
	\item 
$|{\psi}\rangle\triangleleft(B,t_b)$, and 
	\item 
$|{\psi(t_b)}\rangle=\hat U(t_b,t_a)|{\psi(t_a)}\rangle$.
\end{enumerate}

This condition does not contradict the Schr\"odinger equation, and in fact proposes that it remains true even in the cases when we can only think that it is violated by a discontinuous collapse. For this to be true, it is necessary that events like emission and absorption to be weakly localized, so that there is always a solution of the Schr\"odinger equation which satisfies them.

The ``unitary collapse'' condition can be generalized to more particles, and to more places where the particles have to be found. Also, it can be generalized to conditions that are closed not to a particular eigenstate of the position, but of any other observable. This generalization will be made in section \S\ref{s:wavefunction-events-picture}.

\section{Delayed initial conditions}
\label{s:delayed_initial_conditions}

The dependence on the final conditions seems retrocausal, because the initial conditions of the observed system have to be tuned precisely so that the wavefunction becomes localized when its position is detected. But this should not come as a surprise, because we already know that the state of the wavefunction prior to the measurement is constrained by the experimental setup. This is unavoidable in any interpretation in which the outcomes of the measurements are encoded in one way or another in the initial conditions \cite{Bel66,Bel64,KochenSpecker1967HiddenVariables}. A unitary evolution attempt to describe the measurement makes no exception \cite{Sto12QMb}.

The kind of special initial conditions which allow unitary evolution to be compatible with measurements, proven to be required in \cite{Sto12QMb} and used in section \S\ref{s:particle_place2place}, can be interpreted as \emph{superdeterminism} (see for example \cite{hooft2011wave}), or \emph{retrocausality}. This is a delicate problem, because seems to be a threat to the principle of causality. In the following, I will discuss some proposals, and argue that it will not lead to breaking of causality.

This apparently retrocausal feature of quantum mechanics is actually often encountered and discussed in the literature. It is at the origin of the \emph{transactional} \cite{cramer1986transactional,cramer1988overview} and the \emph{time symmetric} \cite{aharonov1964time,aharonov1988result,aharonov1991complete,aharonov2007newinsights,aharonov2007TSV} approaches to QM. Several proposals to deal with this issue are known, for example \cite{deBeauregard1953-DEBMQ,Rietdijk1978retroactiveInfluence,price2008toyRetrocausality,price2015disentangling}.

Here I will argue that the apparent retrocausality can exist in the proposed model without breaking the principle of causality. I will discuss two equivalent pictures, one which is temporal, and another one which is timeless, based on the block view.

The temporal interpretation of the apparent retrocausality is based on \emph{delayed initial conditions} \cite{Sto08f,Sto12QMc,Sto13bSpringer}, in the following sense. The initial conditions of a classical system are usually not restricted -- the system can start in any initial state, and the dynamics will work without problems. The initial state can be measured so that we get the complete information. However, the initial conditions of a quantum system can be seen as not determined until the complete information can be extracted from measurements. The complete information is hidden by the very principles of QM. Therefore, the information usually contained in the initial conditions is distributed in spacetime at the various places where quantum measurements and observations are performed, and no matter how many measurements we perform, we will never find the complete wavefunction of the world. What we can have is a set of possible solutions of the Schr\"odinger equation, which satisfy the observations in a global and self-consistent manner. This set of possible solutions is reduced in time, as new measurements are performed, so that in time we accumulate more and more knowledge about the quantum state. This is not merely a collection of information about the quantum state, because different choices of the observables lead to different possible solutions.
This picture does not violate causality, because it cannot be used to change the past already recorded by observations. The reason is that after each observation we keep only the solutions compatible with the outcome of that observation. So no contradiction is allowed. Of course, the big question is whether the set of solution satisfying all constraints due to observations is always non-empty, no matter how many observations we make. This problem is addressed in this article, by using the fact that there is a trade-off between error and disturbance. In particular, this point is addressed in section \S\ref{s:measurements-as-events} for a specific example.

The timeless picture, which is equivalent to the delayed initial conditions picture, will be discussed in the following.

\section{Spacetime locality}
\label{s:spacetime_locality}

Consider Bohm's version of the Einstein-Podolsky-Rosen experiment (EPR-B) \cite{EPR35,Bel64}. This version involves the entanglement of the spin states of two particles. The analysis of a way by which the EPR-B experiment can take place by unitary evolution was discussed in \cite{Sto08b,Sto08f,Sto12QMc}.

The EPR-B experiment involves the decay of a composite particle which is in a singlet state $\frac {1} {\sqrt{2}}\left(|{\uparrow}_A\rangle|{\downarrow}_B\rangle - |{\downarrow}_A\rangle|{\uparrow}_B\rangle\right)$. The two particles labeled by $A$ and $B$ resulting from the decay arrive at Alice and Bob. Alice measures the spin of the particle $A$ along a direction in space, and Bob measures the spin of $B$. Because both of them find definite and separate outcomes for their experiments, it follows that if unitary evolution is maintained, the two particles arrived at them as separate states \cite{Sto08b}. If we apply backwards in time the evolution equation, we can conclude that after the decay the particles had separate states. This means that between the decay and the measurement both particles behaved locally. Therefore, the correlations between the values obtained by Alice and Bob are enforced locally, through the histories of the two particles. We find again that the states of the particles immediately after the emission had to be fine tuned so that Alice and Bob find them in the correct states.

An experiment verifying what happens in the EPR-B experiment with the weak values of the spin between the emission and the detection of the particles was explored in \cite{aharonov2012future-past}. The conclusion of the article was that 

\begin{quote}
what appears to be nonlocal in \emph{space} turns out to be perfectly local in \emph{spacetime}.
\end{quote}

Although in \cite{aharonov2012future-past} the result is interpreted in terms of the two-state vector formalism (see section \S\ref{s:time_symmetry}), it is also consistent with other interpretations \cite{deBeauregard1953-DEBMQ,Rietdijk1978retroactiveInfluence,price2015disentangling}. In addition, it supports the proposal of this article, that the evolution between the emission and the detection is unitary, and there is no discontinuous collapse.
We can consider the processes taking place during the EPR-B experiment as being local in the sense that the particles are described by local solutions of the Schr\"odinger equation. This kind of spacetime locality is not what we usually expect when we speak about locality, because it depends on the final conditions imposed by the experimental setup. The solutions are local in the sense that they obey partial differential equations on spacetime, but they are also subject to boundary conditions which are global and impose the apparent (space) nonlocality like that from Bell's theorem.

It is normally considered that the Schr\"odinger equation predicts that the particles are entangled after the decay. However, by requiring the solution to satisfy the final conditions, the particles turn out to be separated right after the decay. Since the entire past history of the particles has to satisfy the final conditions, the projection has to be done to the entire life span of the particles, that is, it applies to the past history. This kind of projection which applies to the entire history does not introduce a discontinuity or a violation of the Schr\"odinger equation.

\section{Global consistency condition}
\label{s:global_consistency}

The solutions satisfying both the initial conditions and the final ones are local, as solutions of the Schr\"odinger equation, but they are also subject to global constraints, given by the initial and final constraints, resulting from the preparation and the measurement.

The idea to impose global constraints to local solution is not unprecedented: Schr\"odinger derived the discrete energy spectrum of the electron in the atom by imposing boundary conditions on the sphere at infinity \cite{Sch26}. So, the solutions are local, but among all local solutions, we accept as physical the ones that are consistent everywhere, including at infinity. More generally, they also have to remain consistent in the future. Such conditions are imposed by future measurements, so in order to ensure consistency, we keep only the solutions of the Schr\"odinger equation which are consistent and remain consistent anywhere in spacetime.

The consistency of the solutions with future measurements implies that the state of the system before measurement depends on the observables we will choose to measure in the future \cite{Sto12QMb}, and this has the unpleasant appearance of a conspiracy. This can be interpreted in a less striking way, if we appeal to the \emph{block world} view. The block world view of the universe is mostly known from Einstein's relativity, but it is also useful in Galilean relativity. If we consider the solutions not given by complete initial conditions at some point in the past, but as a combination of delayed initial conditions imposed at various points in spacetime, then the block world view provides a more natural picture \cite{Sto12QMa,Sto12QMc,Sto13bSpringer}.

Other proposals in the same spirit are known, see for example the toy model of using the block world view in \cite{price2008toyRetrocausality}. The resemblance with the toy model proposed in \cite{price2008toyRetrocausality} consists in the fact that both proposals require consistency between conditions imposed at different places and times, in a block world. The difference is that, while the toy model is a simple graph (nevertheless having the desired features of retrocausality), the model proposed in this article is quantum, and is governed by the Schr\"odinger equation, with minimal differences from standard QM (namely, the unitary account of the apparent collapse).

Another way to see this block world picture is as a sheaf of local solutions, which can be combined only in certain ways to obtain a globally consistent solution \cite{Sto12QMc}. Thus, quantum reality is like a puzzle which can be solved only in consistent ways \cite{Sto13bSpringer}.

Even if the block world view may be satisfactory for some aspects of the problem, when we think of the same phenomena in terms of time evolution, the conspiracy and the apparent retrocausality return.
In sections \S\ref{s:wavefunction-events-picture}--\S\ref{s:history_collapse} I will present a more rigorous picture which will show that this does not imply a violation of causality, because it does not change the past, it only determines the parts of the past which were not already determined.

\section{The \emph{wavefunction-events} picture}
\label{s:wavefunction-events-picture}

We are in position to provide a picture of a system or of the Universe, exclusively based on the constraints imposed to the wavefunction by various \emph{events} like emission, detection in a particular (always approximate) position or eigenstate of another observable, passing through slits \textit{etc}.

We denote by $\mathcal M$ the spacetime. It will be useful to define on $\mathcal M$ a time coordinate $t:\mathcal M\to\mathbb{R}$ which foliates it in spacelike surfaces of constant time, $\mathcal M=T\times S=\bigcup_{t\in T} S_t$, where $T$ is an interval in $\mathbb{R}$, and $S_t=\{t\}\times S$, $S$ being the physical space. The following can be applied also to relativistic theories, because it will not break the Lorentz invariance.

Let $\mathcal{H}$ be the total Hilbert space of the universe, which may contain the Fock spaces of all particles, or any suitable space needed to represent the entire universe and all interactions.
We assume that the Schr\"odinger equation
\begin{equation}
\label{schrod_total}
\imath\hbar\frac{\partial}{\partial t}\psi(t)=\hat H \psi(t)
\end{equation}
has solutions of the form $\psi:T\to\mathcal{H}$. We denote the space of solutions of equation \eqref{schrod_total} by
$\mathscr{H}$. 
We denote by $\mathbb P(\mathscr{H})$ the space of solutions viewed as time evolving rays in the Hilbert space $\mathcal{H}$, or time evolving elements of the projective Hilbert space $\mathbb P(\mathcal{H})$.

\begin{definition}
\label{def:event_temporal}
An \emph{event} is a pair $\varepsilon=(t,s)$, where $t\in T$ is a moment of time, and $s$ is a subset of the \emph{projective Hilbert space} $\mathbb P(\mathcal{H})$ of the total system.
\end{definition}

Equivalently, we can take $s$ to be a subset of the Hilbert space of the total system $s\subset\mathcal{H}$, which satisfies the condition that for any $|{\psi}\rangle\in\mathcal{H}$ and any $\alpha\in\mathbb{C}\setminus\{0\}$, $|{\psi}\rangle\in s\Leftrightarrow\alpha|{\psi}\rangle\in s$. This condition just ensures that the set $s$ contains rays from the Hilbert space. 

\begin{example}
\label{ex:event_localized}
For instance, the event that the wavefunction of a particle $|{\psi}\rangle$ is localized at the time $t$ in the region $A$, hence satisfies \eqref{eq_localization_region}, is $\varepsilon=(t,s)$, where $s$ is the set $s=\{|{\psi}\rangle\in\mathcal{H}_1||{\psi}\rangle\triangleleft(A,t)\}\otimes\mathcal{H}_2$, where $\mathcal{H}_1$ is the Hilbert space of the particle, and $\mathcal{H}_2$ is the Hilbert space of the rest of the universe, hence $\mathcal{H}=\mathcal{H}_1\otimes\mathcal{H}_2$.
\end{example}

\begin{definition}
We call an event as in the example \ref{ex:event_localized} a \emph{spacelike event}. Here, the term ``spacelike'' reflects the fact that the points in the region $A$ have the same time, at least in a reference frame, similarly to the case in relativity.
\end{definition}

We see from Example \ref{ex:event_localized} that although the definition of the event \ref{def:event_temporal} refers to the entire Hilbert space, the event itself can be about any subsystem. In this example it was about localization in space, hence around a position eigenstate, but it can be as well about the localization around any state vector, which is an eigenstate of another observable.

We could have taken the definition of an event to be such that $s$ is a Hilbert subspace. This would have allowed us to use projectors. There is a reason why I prefer a general subset and not a Hilbert subspace: not any condition can be expressed by a projector, or a Hilbert subspace. For instance, the set of functions satisfying the condition \eqref{eq_localization_region} does not form a vector space. Hence, using Hilbert spaces instead of sets in Definition \ref{def:event_temporal} would restrict too much the conditions, making the notion of event too simple and unrealistic.

Any position measurement of a particle happens around a particular position and time, but in general that region is extended in spacetime. 
The event that a photon passes through a slit cannot be a spacelike event as in Example \ref{ex:event_localized}, because the slit is also extended in time, and the photon can pass through the slit at various times. For this reason, we need to consider regions $A\subset\mathcal M$ which may also extend in time. But a subset $s$ from an event $\varepsilon=(t,s)$ being a general subset of the projective Hilbert space, allows this situation too. This is because the condition $|{\psi(t)}\rangle\in s$ is equivalent to the condition $|{\psi(t')}\rangle\in \hat U(t',t)s$. In fact, the condition can be seen as selecting a subset of solutions defined for all times, $\psi\in\epsilon$, where $\epsilon\subset\mathscr{H}$ is a subset of the space of solutions $\mathscr{H}$, and not of the Hilbert space $\mathcal{H}$ at a particular time $t$.
This justifies the alternative formulation of Definition \ref{def:event_temporal}:

\begin{definition}
\label{def:event_timeless}
A \emph{timeless event} is a subset $\epsilon\subset\mathbb P(\mathscr{H})$. 
\end{definition}

For any $t$ the timeless events of the form $\epsilon\subset\mathbb P(\mathscr{H})$ are in one-to-one correspondence with the events of the form $\varepsilon=(t,s)$. When time has to be contained explicitly, we will use events of the form $\varepsilon=(t,s)$ as in Definition \ref{def:event_temporal}. This will be in most cases, because the events will be ordered in time. But in reality they can as well be defined in a time independent way, as subsets $\epsilon\subset\mathbb P(\mathscr{H})$, as in Definition \ref{def:event_timeless}.
Let us call the way to represent events from Definition \ref{def:event_temporal} \emph{temporal picture}, and that from Definition \ref{def:event_timeless} the \emph{timeless} picture of events.

At every time there is a collection of possible wavefunctions satisfying all events that already happened, just by the unitary evolution given by the Schr\"odinger equation. Every new event, say every new measurement, only eliminates some of these solutions of the Schr\"odinger equation, but this elimination applies on the entire time range for each solution. By this, it reduces the set of possible solutions from $\mathscr{H}$ to the subset of solutions satisfying also the new events, and giving the appearance of a collapse.

A \emph{registry of events} is a collection of events
\begin{equation}
\mathcal E\subset 2^{\mathbb P(\mathscr{H})}
\end{equation}
or, if the time is specified for each event,
\begin{equation}
\mathcal E\subset T\times 2^{\mathbb P(\mathcal{H})}
\end{equation}
where $2^X$ is the usual notation for the collection of all subsets of a set $X$.

We denote by $\mathscr{H}(\mathcal E)$ the set of the solutions of the Schr\"odinger equation satisfying the registry of events $\mathcal E$.
It is trivial to see that the set of solutions from $\mathscr{H}$ satisfying the events in a registry $\mathcal E$ is,
in terms of events of the form $(t,s)$, is given for any chosen initial time $t_0$ by
\begin{equation}
\label{eq:registry_temporal}
\mathscr{H}(\mathcal E) = \left\{\psi\in\mathscr{H}\left||{\psi(t_0)}\rangle\in\bigcap_{(t,s)\in\mathcal E}\hat U(t_0,t)s\right\}\right.
\end{equation}
or simply
\begin{equation}
\label{eq:registry_timeless}
\mathscr{H}(\mathcal E)=\bigcap_{\epsilon\in\mathcal E}\epsilon
\end{equation}

In most reasonable cases the condition defining an event is about localization in space: a particle was emitted or absorbed or passed through a certain region of spacetime. Also, in the majority of situations the region can be approximated by a spacelike region, hence the event can be a spacelike event as in Example \ref{ex:event_localized}. It is not mandatory for the region $A$ to be connected. If $A$ is the union of two or more connected components, then the particle has two or more alternatives, and it will use all of them. This is the case of the two-slit experiment: the two slits can be seen as two connected components of one single region $A$.

While both of the definitions of events \ref{def:event_temporal} and \ref{def:event_timeless} refer to the Hilbert space of the entire world, the notion of event can be also about subsystems or particles, as already pointed out in Example \ref{ex:event_localized}. Simply take $s$ to be of the form $\mathbb P( s_1\otimes\mathcal{H}_2)$, where $s_1\in\mathcal{H}_1$ is a subset of the Hilbert space $\mathcal{H}_1$ of a subsystem, and $\mathcal{H}_2$ is the Hilbert space of the rest of the world. This allows us to refer when defining an event specifically to subsystems. It works well also for entangled systems.

What about interactions involving some particles as input, and some as output? Such an event can be described in terms of spacelike events. For example, if the region $A$ is bounded by the times $t_1$ and $t_2$, we can describe it by spacelike events at $t_1$ ensuring that the input particles enter the region, and other spacelike events at $t_2$, for the output particles, which ensure that they leave it. We can also add the condition that a particle was annihilated in the region $A$ and therefore does not leave it, or that was created and therefore is not among the input particles, simply by imposing that it does not pass through $S_{t_2}$, respectively $S_{t_1}$.

We can thus use events to describe and construct all sorts of quantum phenomena and experiments.

\section{The history and the wavefunction collapse}
\label{s:history_collapse}

We derive some properties.

\begin{lemma}
\label{thm:properties_events}
For two registries of events $\mathcal E_1$ and $\mathcal E_2$,
\begin{enumerate}
	\item 
$\mathscr{H}(\mathcal E_1)\cap \mathscr{H}(\mathcal E_2)=\mathscr{H}(\mathcal E_1\cup \mathcal E_2),$
	\item 
$\mathscr{H}(\mathcal E_1) \cup \mathscr{H}(\mathcal E_2)\subset\mathscr{H}(\mathcal E_1\cap \mathcal E_2),$
	\item 
If $\mathcal E_1\subseteq \mathcal E_2$, then $\mathscr{H}(\mathcal E_2)\subseteq \mathscr{H}(\mathcal E_1)$.
\end{enumerate}
\end{lemma}
\begin{proof}
We apply Definition \ref{def:event_timeless}, equation \eqref{eq:registry_timeless}, and the properties of operations with sets.
\begin{enumerate}
	\item From equation \eqref{eq:registry_timeless},
\begin{equation}
\begin{array}{l}
\mathscr{H}(\mathcal E_1\cup \mathcal E_2)=\bigcap_{\epsilon\in\mathcal E_1\cup \mathcal E_2}\epsilon \\
= \left(\bigcap_{\epsilon\in\mathcal E_1}\epsilon\right) \cap \left(\bigcap_{\epsilon\in\mathcal E_2}\epsilon\right)=\mathscr{H}(\mathcal E_1)\cap \mathscr{H}(\mathcal E_2)
\end{array}
\end{equation}
	\item From equation \eqref{eq:registry_timeless},
\begin{equation}
\begin{array}{l}
\mathscr{H}(\mathcal E_1) \cup \mathscr{H}(\mathcal E_2) = \left(\bigcap_{\epsilon_1\in\mathcal E_1}\epsilon_1\right) \cup \left(\bigcap_{\epsilon_2\in\mathcal E_2}\epsilon_2\right) \\
= \bigcap_{\epsilon_1\in\mathcal E_1, \epsilon_2\in\mathcal E_2}\left(\epsilon_1\cup \epsilon_2\right) \\
\subseteq \bigcap_{\epsilon_1, \epsilon_2\in\mathcal E_1\cap\mathcal E_2}\left(\epsilon_1\cup \epsilon_2\right) \\
\subseteq \bigcap_{\epsilon\in\mathcal E_1\cap\mathcal E_2}\left(\epsilon \cup \epsilon \right)= \mathscr{H}(\mathcal E_1\cap \mathcal E_2)
\end{array}
\end{equation}
	\item 
If $\mathcal E_1\subseteq \mathcal E_2$, then any solution $\psi\in\mathscr{H}(\mathcal E_2)$ satisfies the events of $\mathcal E_1$, hence $\psi\in\mathscr{H}(\mathcal E_1)$.
Or,
\begin{equation}
\begin{array}{l}
\mathscr{H}(\mathcal E_2) = \mathscr{H}\left(\mathcal E_1\cup (\mathcal E_2\setminus \mathcal E_1)\right) \\
= \mathscr{H}(\mathcal E_1)\cap \mathscr{H}\left(\mathcal E_2\setminus \mathcal E_1\right) \subseteq \mathscr{H}(\mathcal E_1)
\end{array}
\end{equation}
\end{enumerate}
\end{proof}

Consider the spacetime $\mathcal M$, and a registry of events $\mathcal E\subset T\times 2^{\mathbb P(\mathcal{H})}$. From the point of view of a time coordinate, the spacetime $\mathcal M$ is foliated, and the events can be ordered by time. Some of them may be simultaneous with respect to that foliation. For any time $t$, we define the subregistry
\begin{equation}
\mathcal E(t)=\{(t',s)\in\mathcal E|t'\leq t\}
\end{equation}
of events that already passed at the time $t$. For a sequence of times $\ldots < t_{-1} < t_0 < t_1 \ldots$, there is a sequence of event sets
\begin{equation}
\ldots \subseteq \mathcal E_{-1} \subseteq \mathcal E_0 \subseteq \mathcal E_1 \ldots
\end{equation}
where $\mathcal E_i:=\mathcal E(t_i)$, which describes the history of the system. The corresponding sets of solutions of the Schr\"odinger equations also form a sequence (see fig. \ref{Fig1})
\begin{equation}
\ldots \supset \mathscr{H}(\mathcal E_{-1}) \supset \mathscr{H}(\mathcal E_0) \supset \mathscr{H}(\mathcal E_1) \supset \ldots
\end{equation}

\begin{figure}[t!]
\centering
\includegraphics[width=84mm]{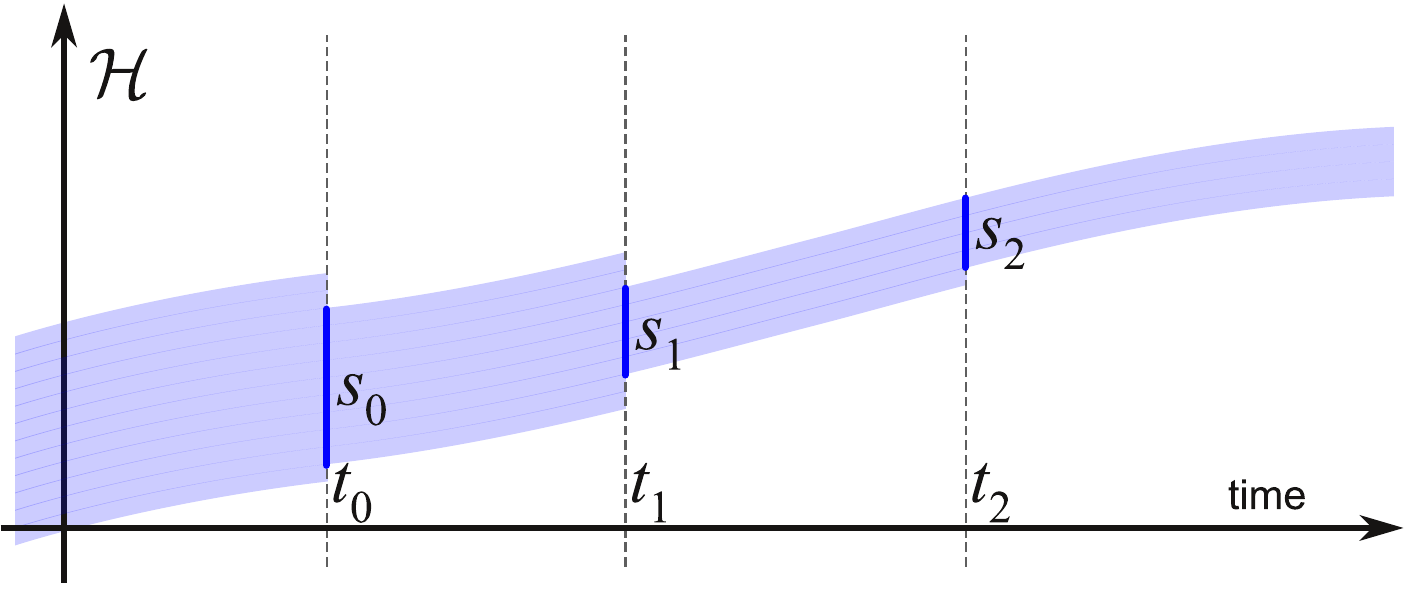}
\caption{\label{Fig1}Delayed initial conditions. Each event reduces the set of solutions compatible with the already happened events.}
\end{figure}

Therefore, after each event the set of solutions satisfying the already passed events is reduced to a subset, similarly to a projection. This takes the appearance of the wavefunction collapse, but the reduction is not real. Rather, we eliminate the solutions that do not correspond to the observations encoded in the events.

Every new event adds new constraints, acting as delayed conditions on the set of admissible solutions of the Schr\"odinger equation.
The global consistency condition is thus satisfied by the reduced set of solutions.

\section{Measurements as events}
\label{s:measurements-as-events}

Is it possible to approximate well enough a history based on unitary evolution interrupted by collapse events (as in von Neumann's formulation) with a history based on unitary evolution only, as in the wavefunction-events picture? In this case, which one is closer to reality and which is an approximation?

Let $\mathcal{H}_1$ be the Hilbert space representing the observed subsystem, and $\mathcal{H}_2$ the Hilbert space of the rest of the world.
A measurement in von Neumann's scheme is accompanied by events saying that certain outcome was obtained. So the possible events corresponding to the measurement of the subsystem represented by the Hilbert space $\mathcal{H}_1$ with the observable $\mathcal O_1$ having as eigenspaces $\mathcal{H}_{1,\mathcal O_1,\lambda}$ for each eigenvalue $\lambda$ are of the form $\varepsilon(t,\mathcal{H}_{1,\mathcal O_1,\lambda}\otimes\mathcal{H}_{2})$.

If we refer only to the subsystem $\mathcal{H}_1$, then two successive incompatible measurements correspond to inconsistent events, in the sense that they cannot be satisfied by the same solution. This requires a discontinuity, so that we have a solution which satisfies the first event, and another one satisfying the second one. Hence the collapse. If we include the ``environment'' $\mathcal{H}_2$, this may contain interactions which change the first solution into the second in a continuous way, and so that for the entire system the evolution is unitary.
This is possible at least in some cases, as we can see from the example of the particle moving from one place to another discussed in section \S\ref{s:particle_place2place}.
But for this to be possible, the initial conditions of both the observed system $\mathcal{H}_1$ and the rest of the world $\mathcal{H}_2$ have to be very special \cite{Sto12QMb}.
Equally important, error is necessary to allow two successive measurements to be compatible.
The example of the photon moving between two locations discussed in section \S\ref{s:particle_place2place} has both special initial conditions, and error, because it is never truly an eigenstate of the position operator.

Can this work for other observables too? Consider for example a particle of spin $\frac 1 2$. Measuring the spin at $t_a$ along the $x$ axis, represented by the observable denoted by $S_x$, can result in two possible outcomes, $|{\uparrow_x}\rangle$ and $|{\downarrow_x}\rangle$. Suppose that the outcome is $|{\uparrow_x}\rangle$. If the particle evolves freely, a subsequent measurement at $t_b>t_a$ along the $z$ direction results in the two possible outcomes $|{\uparrow_z}\rangle$ and $|{\downarrow_z}\rangle$, each of them with probability $\frac 1 2$, since $|{\uparrow_x}\rangle = \frac{\sqrt 2}2\left(|{\uparrow_z}\rangle+|{\downarrow_z}\rangle\right)$. How is it possible that the particle evolves freely from an eigenstate of the observable $S_x$ to an eigenstate of $S_z$?

In the following I will argue that this may be possible even under the assumption of unitary evolution. First, it is clear that measurements, in particular spin measurements, are subject to errors. The experimental setup is such that the position of the detected particles, from which we can infer the spin, are subject to errors. While they can be separated at will in order to exponentially reduce the overlap, there is still error in the orientation of the magnetic field, and approximation of the alignment of the magnetic moment. Like any measurements, we do not actually detect unequivocally the eigenstates, because real measurements are only approximately projective, and actually correspond to POVMs which allow in fact, with a small but nonzero probability, the occurrence of any possible outcome. In addition, in the case of successive measurements of the spin along different axes, the trade-off between error and disturbance can be such that the conditions imposed by both measurements are satisfied. Also, the magnetic fields of the two Stern-Gerlach devices utilized to measure the spin rotate the spin orientation, and this can be such that the result $|{\uparrow_z}\rangle$ or $|{\downarrow_z}\rangle$ is obtained, even if previously the spin was detected to be along a different axis. Moreover, one should not forget that the measurement device itself is a quantum system, but we do not know its complete quantum state. This means that its initial conditions are not completely fixed by our observations, and they introduce some freedom, which may allow it to interact with the observed system such that it disturbs it to lead it into one of the possible eigenstates \cite{Sto08b,Sto12QMc}. All these factors provide enough freedom from the constraints, so that the possibility that unitary evolution is compatible with the outcomes of two successive and non-commuting measurements cannot be easily excluded. To completely reject the possibility of unitary evolution, one should perform successive spin measurements in which we eliminate all possible loopholes which can lead to the necessary disturbance which allows it.
The necessity of special initial conditions of the observed system and the measurement device is visible if we consider the possibility to delay the choice of the second observable, for instance by randomly choosing between the observables $S_x$ and $S_y$. This is because the way the magnetic field of the first Stern-Gerlach device disturbs the observed system after the first measurement has to depend on the orientation of the Stern-Gerlach device performing the second measurement: if the second observable is again $S_x$, then there should not be a disturbance, while if it is $S_z$, the disturbance should be maximal.

For this reason, we should replace the events of the form ``the observed system is in an eigenspace of this observable'' with some more flexible approximations. We would want to have a distance which, when under certain value, tells us that a solution is close enough to a certain subspace of the Hilbert space to be considered the state of the observed system when measured. This approximation should be small enough to be within the experimental error, but large enough to allow, together with the disturbance, for a unitary solution. Up to this point, we do not have a proof that a unitary solution, or at least a continuous one, is always ensured.

\section{Unitary histories}

It is often claimed that the evolution is unitary even during measurements, in the many worlds interpretation (MWI) \cite{Eve57,dWEG73,Vaidman2002MWI}, the \textit{consistent histories} interpretation \cite{Gri84,GH90a,Omn92}, and in the \textit{decoherence program} \cite{Zeh96,Zur98,Zur03a}. In fact, in all these interpretations the unitary evolution is recovered only when considering all the worlds/branches/histories together. At the level of each branch, there is always a collapse. In the decoherence program, the density matrix becomes diagonal, and then it is interpreted as a statistical ensemble, so the measurement reveals that the system was in one of the eigenstates. But if we evolve backwards in time the eigenstate obtained by measurement, we find that the initial state was different than what we considered it to be before diagonalizing the density matrix. Interpreting the diagonalized density matrix as representing a statistical ensemble would solve the problem by unitary evolution, except that, if the chosen observable would have been different, the decomposition of the density matrix as a statistical ensemble would have been completely different. So we have to choose: either admit that there is a discontinuous collapse, or that the initial conditions of the observed system depend on what we will choose to measure \cite{Sto12QMb}. If we stick with unitarity at the level of each branch, we have to admit the solution proposed in this paper.

The solution proposed in section \S\ref{s:wavefunction-events-picture} may be seen as being based on branches that decohere, or worlds that split, but in a different way. Rather than having a unique past that splits in many alternative futures, the split happens for the entire history, as if the past history was precisely the one leading to what we observe in the present. The name \emph{relative state interpretation} may be more appropriate for this interpretation than for the usual MWI. In the wavefunction-events picture, consider a registry $\mathcal E$. A new measurement results in its extension, but the extension depends on the outcome, for example, on the place where a particle was detected. Suppose that the alternatives are described by a collection of events $\epsilon_1$, $\epsilon_2$, \textit{etc}. Then, each alternative event $\epsilon_i$ leads to an alternative extension $\mathcal E_i=\mathcal E\cup\{\epsilon_i\}$. Consequently, the associated set $\mathscr{H}(\mathcal E)$ splits in the alternatives $\mathscr{H}(\mathcal E_i)$. In this sense, the unitary interpretation proposed here can be seen as a many worlds interpretation in which the evolution is actually unitarity for every possible history or branch.

\section{Time symmetry and retrocausality}
\label{s:time_symmetry}

Let us reverse the time in the conditions from section \S\ref{s:unitary_collapse_condition}. We define $|{\psi'(t)}\rangle:=|{\psi(-t)}\rangle$ and $\hat U'(-t_a,-t_b) := \hat U^\dagger(t_b,t_a)$. Then $-t_b < -t_a$, and

\begin{enumerate}
	\item 
The initial constraint becomes
$|{\psi'}\rangle\triangleleft(B,-t_b)$,
	\item 
the final constraint becomes
$|{\psi'}\rangle\triangleleft(A,-t_a)$,
	\item 
$|{\psi'(-t_a)}\rangle=\hat U'(-t_a,-t_b)|{\psi'(-t_b)}\rangle$.
\end{enumerate}

Hence the proposed description is manifestly time symmetric.

To ensure that a wavefunction evolves so that subsequently it becomes localized, its initial conditions have to anticipate the experimental setup from the future \cite{Sto12QMb}. There are other cases where this situation was accepted. For example, the \emph{absorber theory} by Wheeler and Feynman proposed a similar feature in electrodynamics \cite{WheelerFeynman1945AbsorberTheory,WheelerFeynman1949ClassicalElectrodynamicsAbsorber}.
The Lagrangian formulation of QM is also time symmetric, and led to the \emph{sum over histories} approach \cite{Feynman1948SpaceTimeApproachToQM,FeynmanHibbs1965QMAndPathIntegrals}. Another formulation based on Lagrangian, which is also time symmetric, was proposed in \cite{Wharton2007TimeSymmetricQM}.
The \emph{transactional interpretation} of Quantum Mechanics relies on a transaction between past and future \cite{cramer1986transactional,cramer1988overview}.

The \textit{two-state vector formalism} \cite{aharonov1964time,aharonov1988result,aharonov1991complete,aharonov2007newinsights,aharonov2007TSV} also adopts a time symmetric description of quantum mechanics based on a state vector evolving towards the future, and another one towards the past.
In combination with \emph{weak measurements}, this approach turned out to be a powerful tool in identifying and elucidating various quantum paradoxes. In \cite{aharonov2012future-past} is presented a version of the EPR-B experiment which shows how future strong measurements appear to affect the results of weak measurements performed in the past. Moreover, this approach provides important clarifications on the measurement problem and the wavefunction collapse problem, and reveals how time reversibility is attainable, under specific conditions \cite{AharonovCohen2014MeasurementCollapse}.

Another approach based on unitary evolution is the \emph{cellular automaton interpretation of QM}, proposed by 't Hooft \cite{tH07,tHooft2014CellularAutomatonInterpretationQM,Elze2014Action4CellularAutomata}, and also leads to apparent conspiracies between the initial conditions.

A quantum measurement acts like a delayed completion of the initial conditions of the observed system. This appears retrocausal, but cannot be used to change the past, only to decide on the values that were not yet observed, or that were hidden. This is similar to the impossibility to use nonlocality to send signals faster than light.
Basically, each measurement adds a new event, which merely reduces the Hilbert space of the wavefunctions to a subspace. This reduction is, as we have seen, not a change of the solutions, neither a discontinuous collapse, but it is rather similar to an increase in information about the observed systems.

To get a less dramatic picture of this apparent backward causality, we can think at the four-dimensional spacetime as already existing, together with the physical states. The solutions of the Schr\"odinger equation have to be self-consistent not only at a local level, but also globally. This global consistency condition should be imposed also to act in spacetime, not only in space, to remove the inconsistent solutions. The remaining ones appear nonlocal, but this is now just an expression of global consistency \cite{Sto12QMc,Sto13bSpringer}.

\section{Possible implications to quantization of gravity}
\label{s:quantum_gravity}

The \emph{semiclassical Einstein equation} is
\begin{equation}
\label{eq:einstein_semiclassical}
G_{ab} + \Lambda g_{ab} = \frac{8\pi G}{c^4} \langle \hat T_{ab} \rangle
\end{equation}
where $G$ is Newton's constant, $c$ the speed of light, $\Lambda$ the cosmological constant, and $G_{ab}$ Einstein's tensor. The expectation value of the stress-energy tensor $\langle \hat T_{ab} \rangle$ can be taken $\langle \hat T_{ab} \rangle=\langle{\psi}|\hat T_{ab}|{\psi}\rangle$ \cite{Moller1962EnergyMomentumQuantum,Rosenfeld1963QuantizationFields}. Other formulations employ instead of the wavefunction $|{\psi}\rangle$ the density matrix or a $C^\ast$-algebra state \cite{Wal94}.

The main arguments against semi-classical gravity come from the impossibility, or at least difficulty to accommodate the wavefunction collapse with the Einstein equation. If we take into account the backreaction, spacetime curvature has to depend on the way matter is distributed, and conversely. But a collapse would mean a discontinuous change in the curvature, which apparently can be used to send signals faster than light \cite{EppleyHannah1977NecessityQuantizeGravitationalField}. Also, a collapse would break the conservation of the stress-energy tensor. In \cite{PageGeilker1981IndirectEvidenceQM} experiments involving superpositions of macroscopically distinct states, having masses whose gravitational field could be measured, were reported. The gravitational field was found to be correlated only with the eigenstate which was detected. According to the authors, this refuted semiclassical gravity, but in the context of the many worlds interpretation of QM \cite{Eve57,Eve73}. The assumption which was refuted was that if $|{\psi}\rangle$ evolves unitarily in the multiverse, gravity should correspond to the superposition $|{\psi}\rangle$, and not to a particular eigenstate $|{\psi}\rangle_\lambda$ obtained after collapse. Since the gravitational field was found to correspond to one eigenstate and not all states in the superpostion, semiclassical gravity in the context of MWI was refuted.

But if reality is accurately described by a wavefunction which evolves unitarily, or at least continuously, without a discontinuous collapse, then these problems no longer appear, and the semiclassical Einstein equation \eqref{eq:einstein_semiclassical} connects consistently General Relativity and Quantum Mechanics. 
This requires, of course, that in equation \eqref{eq:einstein_semiclassical} the wavefunction $|{\psi}\rangle$ is the ontic, and not an epistemic or statistical one.
Making QM and GR compatible this way does not mean that the world obeys semiclassical gravity, only that, if it is necessary to quantize gravity, it is for other reasons.

\section{Open problems}
\label{s:open}

Here I argued for the possibility to avoid a discontinuous collapse by maintaining unitary evolution, or at least continuity, during the apparent collapse. If such a solution can be proven to be consistent, this would resolve the conflict between collapse and dynamics, but it will also make QM and GR consistent semiclassically. This possibility justifies the research in this direction. However, severe difficulties have to be resolved. The first problem is to find what ontic states correspond to each state obtained from measurement. In other words, for each epistemic state resulting from a measurement, what ontic states can lead to the resulting outcome. For example, in section \S\ref{s:particle_place2place} it was shown that the photon cannot be localized at a point, or even in a compact region, if we want to maintain unitarity. This has to be done so that it can apply to all possible observables.  Maybe this approach will fail already at this step by turning out that it is not flexible enough to describe the apparent collapse. If it works, the next necessary step is to deduce the Born rule from the correspondence between epistemic and ontic states, extending the results obtained in \cite{shaarawi1994diffraction} to all cases. 
Can we find experimental evidence supporting the unitary or at least continuous, rather than the discontinuous collapse based QM?
Can we find rigorous theoretical evidence, for example from the consistency between QM and GR required by semiclassical gravity?
At least we have seen that it is possible to save unitarity, and this possibility worth being explored, for its implications to the foundations of QM and of semiclassical gravity.

\subsection*{Acknowledgement}

I thank the Editors for valuable recommendations which improved the clarity and quality of the article.

\end{document}